\documentclass[12pt]{article}
\usepackage{amsmath, amssymb, amsthm, mathrsfs}
\usepackage{enumerate}
\usepackage{color}

\newtheorem{claim}{Claim}[section]
\newtheorem{theorem}[claim]{Theorem}
\newtheorem{lemma}[claim]{Lemma}
\newtheorem{remark}[claim]{Remark}

\newtheorem{example}[claim]{Example}

\usepackage{cite}

\begin{document}
\begin{center}
{\Large{\textbf{SEMICLASSICAL BOUNDS \\[.3em] IN MAGNETIC BOTTLES}}}

\bigskip

{\large{Diana Barseghyan$^{a,b}$, Pavel Exner$^{a,c}$,  Hynek Kova\v{r}\'{\i}k$^{d}$, \\[.3em] Timo Weidl$^{e}$}}

\end{center}

\begin{quote}

{\small a) Department of Theoretical Physics, Nuclear Physics Institute ASCR, \\ \phantom{c)} 25068 \v{R}e\v{z} near Prague, Czech Republic \\[.3em]
 b) Department of Mathematics, Faculty of Science, University \\ \phantom{c)} of Ostrava, 30.~dubna 22, 70103 Ostrava, Czech Republic \\[.3em]
c) Doppler Institute for Mathematical Physics and Applied \\ \phantom{c)} Mathematics, Czech Technical University, B\v{r}ehov\'{a} 7, 11519 Prague \\[.3em]
d) Dicatam, Sezione di Matematica, Universit\`a degli Studi \\ \phantom{c)} di Brescia, via Branze~38, 25123 Brescia, Italy \\[.3em]
e) Fakult\"at f\"ur Mathematik und Physik Institut f\"ur Analysis, \\ \phantom{c)} Dynamik und Modellierung, Universit\"at Stuttgart, Pfaffenwaldring~57, \\ \phantom{c)} 70569 Stuttgart, Germany \\[.5em]
\phantom{c)} \emph{dianabar@ujf.cas.cz, diana.barseghyan@osu.cz, exner@ujf.cas.cz, \\ \phantom{c)}hynek.kovarik@unibs.it, weidl@mathematik.uni-stuttgart.de}}

\end{quote}

\bigskip

\noindent \textbf{Abstract.} The aim of the paper is to derive spectral estimates into several classes of magnetic systems. They include three-dimensional regions with Dirichlet boundary as well as a particle in $\mathbb{R}^3$ confined by a local change of the magnetic field. We
establish two-dimensional Berezin-Li-Yau and Lieb-Thirring-type bounds in the presence of magnetic fields and, using them, get three-dimensional estimates for the eigenvalue moments of the corresponding magnetic Laplacians.

\section{Introduction} \label{s: intro}

\setcounter{equation}{0}
Let $-\Delta_\Omega$ be the Dirichlet Laplacian corresponding to an open bounded domain $\Omega\subset\mathbb{R}^d$, defined in the quadratic form sense on $\mathcal{H}^1_0(\Omega)$. The operator is obviously non-negative and since the embedding $\mathcal{H}_0^1\hookrightarrow L^2(\Omega)$ is compact, its spectrum is purely discrete accumulating at infinity only. It is well known that for $d=3$, up to a choice of the scale, the eigenvalues describe energies of a spinless quantum particle confined to such a hard-wall `bottle'.

Motivated by this physical problem, we consider in the present work a magnetic version of the mentioned Dirichlet Laplacian, that is, the operator $\mathcal{H}_\Omega(A)=(i\nabla+A(x))^2$ associated with the closed quadratic form
$$
\|(i\nabla+A)u\|_{L^2(\Omega)}^2\,,\quad u\in\mathcal{H}_0^1(\Omega)\,,
$$
where the real-valued and sufficiently smooth function $A$ is a vector potential. The magnetic Sobolev norm on the bounded domain $\Omega$ is equivalent to the non-magnetic one and the operator $\mathcal{H}_\Omega(A)$ has a purely discrete spectrum as well. We shall denote the eigenvalues by $\lambda_k =\lambda_k(\Omega, A)$, assuming that they repeat according to their multiplicities.

One of the objects of our interest in this paper will be bounds of the eigenvalue moments of such operators. For starters, recall that for non-magnetic Dirichlet Laplacians
the following bound was proved in the work of Berezin, Li and Yau  \cite{Be72a, Be72b, LY83},
\begin{equation}
\label{Berezin bound}\sum_k(\Lambda-\lambda_k(\Omega,0))_+^\sigma\le L_{\sigma,d}^{\mathrm{cl}}\,|\Omega|\,
\Lambda^{\sigma+\frac{d}{2}} \quad\text{for any}\;\;\sigma\ge1 \;\;\text{and}\;\;\Lambda>0\,,
\end{equation}
where $|\Omega|$ is the volume of $\Omega$ and the constant on the right-hand side,
$$
L_{\sigma,d}^{\mathrm{cl}}=\frac{\Gamma(\sigma+1)}{(4\pi)^{\frac{d}{2}}\Gamma
(\sigma+1+d/2)}\,,
$$
is optimal. Furthermore, the bound (\ref{Berezin bound}) holds true for $0\le\sigma<1$ as well, but with another, probably non-sharp constant on the right-hand side,
\begin{equation}\label{Laptev ineq.}
\sum_k(\Lambda-\lambda_k(\Omega,0))_+^\sigma\le
2\left(\frac{\sigma}{\sigma+1}\right)^\sigma L_{\sigma,d}^{\mathrm{cl}}\,|\Omega|\,
\Lambda^{\sigma+\frac{d}{2}}\,,\quad 0\le\sigma<1\,.
\end{equation}
see \cite{La97}. In the particular case $\sigma=1$ the inequality (\ref{Berezin bound}) is equivalent, via Legendre transformation, to the lower bound
\begin{equation}\label{Leg.trans.}
\sum_{j=1}^N\lambda_j(\Omega,0)\ge C_d|\Omega|^{-\frac{2}{d}}N^{1+\frac{2}{d}}\,,\quad C_d=\frac{4\pi d}{d+2}\Gamma(d/2+1)^{\frac{2}{d}}\,.
\end{equation}

Turning next to the magnetic case, we note first that the pointwise diamagnetic inequality \cite{LL01}, namely
$$
|\nabla|u(x)||\le|(i\nabla+A)u(x)|\quad\text{for a.a.}\;\; x\in\Omega\,,
$$
implies $\lambda_1(\Omega, A)\ge\lambda_1(\Omega,0)$, however, the estimate $\lambda_j(\Omega, A)\ge\lambda_j(\Omega,0)$ fails in general if $j\ge2$. Nevertheless, momentum estimates are still valid for some values of the parameters. In particular, it was shown \cite{LW00} that the sharp bound (\ref{Berezin bound}) holds true for arbitrary magnetic fields provided $\sigma\ge\frac{3}{2}$, and the same sharp bound holds true for constant magnetic fields if $\sigma\ge1$, see\cite{ELV00}. Furthermore, in the dimension $d=2$ the bound (\ref{Laptev ineq.}) holds true for constant magnetic fields if $0\le\sigma<1$ and the constant on its right-hand side cannot be improved \cite{FLW09}.

Our main aim in the present work is to derive sufficiently precise two-dimensional Berezin-type estimates for quantum systems exposed to a magnetic field and to apply them to the three-dimensional case. We are going to address two questions, one concerning eigenvalue moments estimates for magnetic Laplacians on three dimensional domains having a bounded cross section in a fixed direction, and the other about similar estimates for magnetic Laplacians defined on whole $\mathbb{R}^3$.

Let us review the paper content in more details. In Sec.~\ref{s: reduction} we will describe the dimensional-reduction technique \cite{LW00} which allows us to derive the sought spectral estimates for three-dimensional magnetic `bottles' using two-dimensional ones. Our next aim is to derive a two-dimensional version of the Li-Yau inequality (\ref{Leg.trans.}) in presence of a constant magnetic field giving rise to an extra term on the right-hand side. The result will be stated and proved in first part of Sec.~\ref{s: Berezin-Li-Yau}. This in turn will imply, by means of Legendre transformation, a magnetic version of the Berezin inequality which we are going to present in second part of Sec.~\ref{s: Berezin-Li-Yau}. It has to be added that the question of semiclassical spectral bounds for such systems has been addressed before, in particular, another version of the magnetic Berezin inequality was derived by two of us \cite{KW13}. In final part of Sec.~\ref{s: Berezin-Li-Yau} we are going to compare the two results and show that the one derived here becomes substantially better when the magnetic field is strong.

In some cases the eigenvalues of the magnetic Dirichlet Laplacian with a constant magnetic field can be computed exactly in terms of suitable special functions.  In the first part of Sec.~\ref{s: disc} we are present such an example considering the magnetic  Dirichlet Laplacian on a two-dimensional disc with a constant magnetic field. Its eigenvalues will be expressed in terms of Kummer function zeros. Next, in the second part, we are going to consider again the magnetic Dirichlet  Laplacian on a two-dimensional disc, now in a more general situation when the magnetic field is no longer homogeneous but retains the radial symmetry; we will derive the Berezin inequality for the eigenvalue moments. In Sec.~\ref{3Dapplication} we shall return to our original motivation and use the mentioned reduction technique to derive Berezin-type spectral estimates for a class of three-dimensional magnetic `bottles' characterized by a bounded cross section in the $x_3$ direction.

Turning to the second one of the indicated questions, from Sec.~\ref{s:mgBerezin} on, we shall be concerned with magnetic Laplacians in $L^2(\mathbb{R}^3)$ associated with the magnetic field $B:\mathbb{R}^3\to\mathbb{R}^3$ which is as a local perturbation of a constant magnetic field of intensity $B_0>0$.  Again, as before, we first derive suitable two-dimensional estimates; this will be done in Sec.~\ref{s:mgBerezin}. In the last two sections we apply this result to the three-dimensional case. In Sec.~\ref{s:3Dhole} we show that the essential spectrum of the magnetic Laplacian with corresponding perturbed magnetic field coincides with $[B_0, \infty)$. The Sec.~\ref{LT-3D} we then prove Lieb-Thirring-type inequalities for the moments of eigenvalues below the threshold of the essential spectrum for several types of magnetic `holes'.

\section{Dimensional reduction}
\label{s: reduction}
\setcounter{equation}{0}

As indicated our question concerns estimating eigenvalues due to confinement in a three-dimensional `bottle' by using two-dimensional Berezin type estimates. In such situation one can use the dimension-reduction technique \cite{LW00}. In particular, let $-\Delta_\Omega$ be the Dirichlet Laplacian on an open domain $\Omega\subseteq\mathbb{R}^3$, then for any $\sigma\ge\frac{3}{2}$ the inequality
\begin{equation}\label{Dir.Laplacian}
\mathrm{tr}\left(\Lambda-(-\Delta_\Omega)\right)_+^\sigma\le L_{1,\sigma}
^{\mathrm{cl}}\int_{\mathbb{R}}\mathrm{tr}\left(\Lambda-(-\Delta_{\omega(x_3)})\right)_+^{\sigma+
\frac{1}{2}}\,\mathrm{d}x_3
\end{equation}
is valid, where $-\Delta_{\omega(x_3)}$ is the Dirichlet Laplacian on the section
$$
\omega(x_3)=\left\{x'=(x_1, x_2)\in\mathbb{R}^2|\,\,x=(x', x_3)=(x_1, x_2, x_3)\in\Omega
\right\},
$$
see \cite{LW00}, and also \cite{ELW04, Wei08}. The integral at the right-hand side of (\ref{Dir.Laplacian}), in fact restricted to those $x_3$ for which $\inf\,\mathrm{spec}(-\Delta_{\omega(x_3)})<\Lambda$, yields the classical phase space volume. Note that in this way one can obtain estimates also in some unbounded domains \cite{GW11} as well as remainder terms \cite{Wei08}.

A similar technique can be used also in the magnetic case. To describe it, consider a sufficiently smooth magnetic vector potential $A(\cdot):\Omega\to \mathbb{R}^3$ generating the magnetic field
$$
B(x)=(B_1(x),B_2(x),B_3(x))=\mathrm{rot}\,A(x)\,.
$$
For the sake of definiteness, the shall use the gauge with $A_3(x)=0$. Furthermore, we consider the magnetic Dirichlet Laplacians
$$
\mathcal{H}_\Omega(A)=(i\nabla_x-A(x))^2\quad\text{on}\;\, L^2(\Omega)
$$
and
$$
\widetilde{H}_{\omega(x_3)}(\widetilde{A})=(i\nabla_{x'}-\widetilde{A}(x))^2\quad
\text{on}\;\, L^2(\omega(x_3))\,,
$$
where $\widetilde{A}(x):=(A_1(x), A_2(x))$. Note that for the fixed $x_3$ the two-dimensional vector potential $\widetilde{A}(x', x_3)$ corresponds to the magnetic field
$$
\tilde{B}(x', x_3)=B_3(x)=\frac{\partial A_2}{\partial x_1}-\frac{\partial A_1}
{\partial x_2}\,.
$$
Referring to\cite[Sec.~3.2]{LW00} one can then claim that for a $\sigma\ge\frac{3}{2}$ we have
\begin{equation}
\label{magn.field}\mathrm{tr}(\Lambda-\mathcal{H}_\Omega(A))_+^\sigma\le L_{1,\sigma}^{\mathrm{cl}}
\int_{\mathbb{R}}\mathrm{tr}(\Lambda-\widetilde{H}_{\omega(x_3)}(\widetilde{A}))_+^
{\sigma+1/2}\,\mathrm{d}x_3\,.
\end{equation}

\section{Berezin-Li-Yau inequality with a constant magnetic field}
\label{s: Berezin-Li-Yau}
\setcounter{equation}{0}

Suppose that the motion is confined to a planar domain $\omega$ being exposed to influence of a constant magnetic field of intensity $B_0$ perpendicular to the plane, and let $A:\: \mathbb{R}^2\to \mathbb{R}^2$ be a vector potential generating this field. We denote by $H_\omega(A)$ the corresponding magnetic Dirichlet Laplacian on $\omega$ and $\mu_j(A)$ will be its eigenvalues arranged in the ascending with repetition according to their multiplicity. Our first aim is to extend the Li-Yau inequality (\ref{Leg.trans.}) to this situation with an additional term on the right-hand side depending on $B_0$ only. This will be then used to derive the Berezin-type inequality. Conventionally we denote by $\mathbb{N}$ the set of natural numbers, while the set of integers will be denoted by $\mathbb{Z}$.
\medskip

\noindent The following result is not new. Indeed, it can be recovered from \cite[Sec.~2]{ELV00}, however, for the sake of completeness we include a proof.

\subsection{Li-Yau estimate}
\setcounter{equation}{0}
\begin{theorem}
Assume that $\omega\subset\mathbb{R}^2$ is open and finite. Then the inequality
\begin{equation}\label{eigenvalue}
\sum_{j\le N}\mu_j(A)\ge\frac{2\pi N^2}{|\omega|}+
\frac{B_0^2}{2\pi}|\omega| m(1-m)
\end{equation}
holds, where $m:=\left\{\frac{2\pi N}{B_0|\omega|}\right\}$ is the fractional part of $\frac{2\pi N}{B_0|\omega|}$.
\end{theorem}
\begin{proof}
Without loss of generality we may assume that $B_0>0$. Let $P_k$ be the orthogonal projection onto the $k$-th Landau level, $B_0(2k-1)$, of the Landau Hamiltonian $(i\nabla+A(x))^2$ in $L^2(\mathbb{R}^2)$ which is an integral operator with the kernel $P_k(x,y)$ -- see \cite{KW13}. Note that we have
\begin{equation}\label{P_k}
P_k(x,x)= \frac{1}{2\pi}B_0\,,
\end{equation}
\begin{eqnarray}
\lefteqn{\int_{\mathbb{R}^2}\left(\int_\omega|P_k(y,x)|^2\,\mathrm{d}x\right)\,\mathrm{d}y
=\int_\omega\left(\int_{\mathbb{R}^2}P_k(y,x)\,\overline{P_k(x,y)}\,\mathrm{d}y\right)
\,\mathrm{d}x} \nonumber \\ && \label{Landau level}
\qquad\qquad\quad =\int_\omega P_k(x,x)\,\mathrm{d}x=\frac{B_0}{2\pi}|\omega|\,. \phantom{AAAAAAAAAA}
\end{eqnarray}
Let $\phi_j$ be a normalized eigenfunction corresponding to the eigenvalue $\mu_j(A)$.
We put $f_{k,j}(y):=\int_\omega P_k(y,x)\phi_j(x)\,\mathrm{d}x$, where $y\in\mathbb{R}^2$, and furthermore
$$
F_N(k):=\sum_{j\le N}\|f_{k,j}\|_{L^2(\mathbb{R}^2)}^2\,.
$$
We have the following identity,
\begin{eqnarray*}
\lefteqn{\sum_{j\le N}\mu_j(A)=\sum_{j\le N}\|(i\nabla-A)\phi_j\|_{L^2(\omega)}^2}
 \\ &&
=\sum_{j\le N}\sum_{k\in\mathbb{N}}\|(i\nabla-A)f_{k,j}\|_{L^2(\mathbb{R}^2)}^2
 \\ &&
=\sum_{k\in\mathbb{N}}B_0(2k-1)\sum_{j\le N}\|f_{k,j}\|_{L^2(\mathbb{R}^2)}^2
 \\ &&
=\sum_{k\in\mathbb{N}}B_0(2k-1)F_N(k) =:J[F_N]\,.
\end{eqnarray*}
Moreover, the normalization of the functions $\phi_j$ implies
\begin{equation}\label{FN}
\sum_{k\in\mathbb{N}}F_N(k)=
\sum_{j\le N}\sum_{k\in\mathbb{N}}\|f_{k,j}\|_{L^2(\mathbb{R}^2)}^2
=\sum_{j\le N}\|\phi_j\|_{L^2(\omega)}^2=N\,.
\end{equation}
Finally, in view of Bessel's inequality the following estimate holds true,
\begin{eqnarray}
\lefteqn{F_N(k)=\sum_{j\le N}\|f_{k,j}\|_{L^2(\mathbb{R}^2)}^2
=\int_{\mathbb{R}^2}\left|\sum_{j\le N}\int_\omega P_k(y,x)\phi_j(x)\,\,
\mathrm{d}x\right|^2\,\mathrm{d}y} \nonumber \\ &&
\label{tildeB}
\qquad \le\int_{\mathbb{R}^2}
\left(\int_\omega|P_k(y,x)|^2\,\mathrm{d}x\right)\,\mathrm{d}y
=\frac{B_0}{2\pi}|\omega|\,. \phantom{AAAAAAA}
\end{eqnarray}
Let us now minimize the functional $J[F_N]$ under the constraints (\ref{FN}) and (\ref{tildeB}). To this aim, recall first the \emph{bathtub principle}  \cite{LL01}:

\medskip

\noindent Given a $\sigma$-finite measure space $(\Omega,\,\Sigma,\,\mu)$, let $f$ be a real-valued measurable function on $\Omega$ such that $\mu\{x:f(x)<t\}$ is finite for all $t\in\mathbb{R}$. Fix further a number $G>0$ and define a class of measurable functions on $\Omega$ by
$$
\mathcal{C}=\left\{g:0\le g(x)\le1\quad\text{for all}\quad x\quad\text {and}\quad\int_\Omega g(x)\mu(\mathrm{d}x)=G\right\}\,.
$$
Then the minimization problem of the functional
$$
I=\inf_{g\in\mathcal{C}}\int_\Omega f(x)g(x)\mu(\mathrm{d}x)
$$
is solved by
\begin{equation} \label{minimizer}
g(x)=\chi_{\{f<s\}}(x)+c\chi_{\{f=s\}}(x)\,,
\end{equation}
giving rise to the minimum value
$$
I=\int_{\{f<s\}}f(x)\mu(\mathrm{d}x)+cs\mu\{x:f(x)=s\}\,,
$$
where
$$
s=\sup\{t:\mu\{x:\,f(x)<t\}\le G\}
$$
and
$$
c\mu\{x:f(x)=s\}=G-\mu\{x:f(x)<s\}\,.
$$
Moreover, the minimizer given by (\ref{minimizer}) is unique if $G= \mu\{x:f(x)<s\}$ or if  $G=\mu\{x:f(x)\le s\}$.

\medskip

\noindent Applying this result to the functional $J[F_N]$ with the constraints (\ref{FN}) and (\ref{tildeB}) we find that the corresponding minimizers are
$$
F_N(k)=\frac{B_0}{2\pi}|\omega|\,,\quad k=1,2,\ldots,M\,,
$$
$$
F_N(M+1)=\frac{B_0}{2\pi}|\omega| m\,,
$$
$$
F_N(k)=0,\quad k>M+1,
$$
where $M=\left[\frac{2\pi N}{B_0|\omega|}\right]$ is the entire part and $m=\left\{\frac{2\pi N}{B_0|\omega|}\right\}$, so that $M+m=\frac{2\pi N}{B_0|\omega|}$. Consequently, we have the lower bound
\begin{eqnarray*}
\lefteqn{J[F_N]\ge\frac{B_0}{2\pi}|\omega|\sum_{k=1}^M(2k-1)B_0
+\frac{B_0}{2\pi}|\omega|m(2M+1)B_0} \\ &&
=\frac{B_0}{2\pi}|\omega| (M^2+2M m+m) \\ &&
=\frac{B_0^2}{2\pi}|\omega|(M+m)^2+\frac{B_0^2}{2\pi}|\omega|
(m-m^2)
\end{eqnarray*}
which implies
$$
\sum_{j\le N}\mu_j(A)\ge\frac{2\pi N^2}{|\omega|}+
\frac{B_0^2}{2\pi}|\omega| m(1-m)\,.
$$
This is the claim we have set out to prove.
\end{proof}

Since $0\le m<1$ by definition the last term can regarded as a non-negative remainder term, which is periodic with respect to $\frac{N} {\Phi}$, where $\Phi=\frac{B_0|\omega|}{2\pi}$ is the magnetic flux, i.e. the number of flux quanta through $\omega$. Note that for $N<\Phi$ the right-hand side equals $NB$ and for large enough $B_0$ this estimate is better than the lower bound in terms of the phase-space volume.

\subsection{A magnetic Berezin-type inequality}
\setcounter{equation}{0}
The result obtained in the previous subsection allows us to derive an extension of the Berezin inequality to the magnetic case. We keep the notation introduced above, in particular, $H_\omega(A)$ is the magnetic Dirichlet Laplacian on $\omega$ corresponding to a constant magnetic field $B_0$ and $\mu_j(A)$ are the respective eigenvalues. Without loss of generality we assume again that $B_0>0$. Then we can make the following claim.

\begin{theorem} \label{th-berezin}
Let $\omega\subset\mathbb{R}^2$ be open and finite, then for any $\Lambda>B_0$ we have
\begin{equation}\label{Berezin}
\sum_{j=1}^N(\Lambda-\mu_j(A))\le\frac{(\Lambda^2
-B_0^2)|\omega|}{8\pi}+\frac{(\Lambda-B_0) B_0|\omega|}
{4\pi}\left\{\frac{\Lambda+B_0}{2B_0}\right\}.
\end{equation}
\end{theorem}
\begin{proof}
Subtracting $N\Lambda$ from both sides of inequality (\ref{eigenvalue}), we get
\begin{equation}\label{eigenvalue sum}
\sum_{j=1}^N(\Lambda-\mu_j(A))\le N\Lambda-
\frac{2\pi N^2}{|\omega|}-\frac{B_0^2}{2\pi}|\omega|m(1-m)\,,
\end{equation}
and consequently
$$
\sum_{j=1}^N(\Lambda-\mu_j(A))_+\le \left(N\Lambda-\frac{2\pi
N^2}{|\omega|}-\frac{B_0^2}{2\pi}|\omega|m(1-m)\right)_+.
$$
We are going to investigate the function $f:\,\mathbb{R}_+\to \mathbb{R}$,
$$
f(z):=z\Lambda-\frac{2\pi z^2}{|\omega|}
-\frac{B_0^2|\omega|}{2\pi}\left\{\frac{2\pi z}{B_0|\omega|}\right\}
\left(1-\left\{\frac{2\pi z}{B_0|\omega|}\right\}\right),
$$
on the intervals
$$
\frac{B_0|\omega|k}{2\pi}
\le z<\frac{B_0|\omega|(k+1)}{2\pi},\quad k=0,1,2,\ldots,
$$
looking for an upper bound. It is easy to check that
\begin{eqnarray*}
\lefteqn{f'(z)=\Lambda-\frac{4\pi}{|\omega|}z-\frac{B_0^2
|\omega|}{2\pi}\frac{2\pi}{B_0|\omega|}+\frac{2B_0^2|\omega|}{2\pi}
\left\{\frac{2\pi z}{B_0|\omega|}\right\}\frac{2\pi}{B_0|\omega|}}
\\ &&
\quad =\Lambda-\frac{4\pi}{|\omega|}z-B_0+2B_0\left\{\frac{2\pi z}{B_0|\omega|}\right\}, \phantom{AAAAAAAAAAA}
\end{eqnarray*}
thus the extremum of $f$ is achieved at the point $z_0$ such that
\begin{equation}\label{x_0}
\Lambda-B_0-\frac{4\pi}{|\omega|}z_0+2B_0
\left\{\frac{2\pi z_0}{B_0|\omega|}\right\}=0\,.
\end{equation}
Denoting $x_0:=\frac{2\pi z_0}{B_0|\omega|}$, the condition reads $\Lambda-2B_0 x_0-B_0+2B_0 \{x_0\}=0$ giving
$$
x_0=\frac{\Lambda-B_0+2B_0 \{x_0\}}{2B_0}.
$$
It yields the value of function $f$ at $z_0$, namely
\begin{eqnarray}
\lefteqn{f(z_0)=\frac{\Lambda B_0|\omega|}{2\pi}
\frac{(\Lambda-B_0+2B_0\{x_0\})}{2B_0}-\frac{B_0^2|\omega|}
{2\pi}\left(\frac{\Lambda-B_0+2B_0\{x_0\}}{2B_0}\right)^2} \nonumber \\ &&
\quad -\frac{B_0^2|\omega|}{2\pi}\{x_0\}(1-\{x_0\}) \nonumber \\ &&
=\frac{\Lambda|\omega|}{4\pi}
\left(\Lambda-B_0+2B_0\{x_0\}\right)-\frac{|\omega|}{8\pi}\left(\Lambda-
B_0+2B_0\{x_0\}\right)^2 \nonumber \\ && \quad -\frac{B_0^2|\omega|}{2\pi}\{x_0\}(1-\{x_0\}) \nonumber \\ &&
=\frac{|\omega|}{4\pi}\biggl(\Lambda(\Lambda-B_0+2B_0\{x_0\})
-\frac{(\Lambda-B_0+2B_0\{x_0\})^2}{2} \nonumber \\ && \quad -2B_0^2\{x_0\}(1-\{x_0\}) \biggr) \nonumber \\ &&
=\frac{|\omega|}{4\pi}\biggl(\Lambda^2-\Lambda B_0+2\Lambda B_0\{x_0\}-\frac{\Lambda^2}{2}+\Lambda B_0-\frac{B_0^2}{2}
-2\Lambda B_0\{x_0\} \nonumber \\ &&
\quad +2B_0^2\{x_0\}  -2B_0^2
\{x_0\}^2-2B_0^2\{x_0\}+2B_0^2\{x_0\}^2\biggr) \nonumber \\ && =\frac{|\omega|(\Lambda^2-B_0^2)}{8\pi}\,. \label{inequality}
\end{eqnarray}
Furthermore, the values of $f$ at the endpoints $\frac{B_0 k|\omega|}{2\pi},\,k=0,1,2,\ldots\,$, equal
$$
f\left(\frac{B_0 k|\omega|}{2\pi}\right)=\frac{B_0 \Lambda k|\omega|}{2\pi}-\frac{2\pi}{|\omega|}
\frac{B_0^2 k^2|\omega|^2}{4\pi^2}=\frac{B_0 k|\omega|}{2\pi}(\Lambda-k B_0)\,.
$$
Consider now an integer $m$ satisfying $1\le m\le\left[\frac{\Lambda+B_0}{2B_0} \right]$, then
\begin{eqnarray}
\nonumber \lefteqn{f\left(\frac{B_0|\omega|}{2\pi}\left(\left[\frac{\Lambda+B_0}{2B_0}
\right]-m\right)\right)} \\ && \nonumber
=\frac{B_0|\omega|}{2\pi}\left(\left[\frac{\Lambda+
B_0}{2B_0}\right]-m\right)\left(\Lambda-\left(\left[\frac{\Lambda+
B_0}{2B_0}\right]-m\right)B_0\right) \\ && \nonumber
\le\frac{B_0|\omega|}{2\pi}\left(\frac{\Lambda+B_0}{2B_0}-m\right)
\left(\Lambda-\left(\frac{\Lambda+B_0}{2B_0}-m\right)B_0+\left\{
\frac{\Lambda+B_0}{2B_0}\right\}B_0\right) \\ &&
\nonumber=\frac{\left(\Lambda-(2m-1)B_0\right)|\omega|}{4\pi}\left(\frac{\Lambda
+(2m-1)B_0}{2}+\left\{\frac{\Lambda+B_0}
{2B_0}\right\}B_0\right)\\ && \nonumber=\frac{\left(\Lambda^2-(2m-1)^2
B_0^2\right)|\omega|}{8\pi}+\frac{\left(\Lambda-(2m-1)B_0\right)B_0
|\omega|}{4\pi}\left\{\frac{\Lambda+B_0}{2B_0}\right\}\\ && \label{extremum}
\le\frac{(\Lambda^2-B_0^2)|\omega|}{8\pi}+\frac{\left(\Lambda-B_0\right)
B_0|\omega|}{4\pi}\left\{\frac{\Lambda+B_0}{2B_0}\right\}.
\end{eqnarray}
On the other hand, for integers satisfying $k\ge \left[\frac{\Lambda+B_0}{2B_0}\right]$ one can check easily that
$$4
B_0^2 k^2-4B_0\Lambda k+\Lambda^2-B_0^2\ge0\,,
$$
which means
\begin{equation}\label{maximum}
\frac{B_0 k|\omega|}{2\pi}(\Lambda-B_0 k)\le \frac{(\Lambda^2-B_0^2)|\omega|}{8\pi}\,.
\end{equation}
Combining inequalities (\ref{extremum}) and (\ref{maximum}) we conclude that at the interval endpoints, $z=\frac{B_0 k|\omega|}{2\pi},\,k=0,1,2,\ldots\,$, the value of function $f$ does not exceed $\frac{(\Lambda^2-B_0^2)|\omega|}{8\pi}+\frac{(\Lambda-B_0) B_0|\omega|}{4\pi}\left\{\frac{\Lambda+B_0}{2B_0}\right\}$. Hence in view of (\ref{inequality}) we have
$$
f(z)\le\frac{(\Lambda^2-B_0^2)|\omega|}{8\pi}+\frac{(\Lambda-B_0)
B_0|\omega|}{4\pi}\left\{\frac{\Lambda+B_0}{2B_0}\right\}
$$
for any $z\ge 0$. Combining this inequality above with the bound (\ref{eigenvalue sum}), we arrive at the desired conclusion.
\end{proof}

\begin{remark} \label{AL-rmk}
{\rm Using the Aizenman-Lieb procedure \cite{AL78} and the fact that $\inf\sigma(H_\omega(A))\ge B_0$ we can get also bound for other eigenvalue moments. Specifically, for any $\sigma\ge3/2$ Theorem~\ref{th-berezin} implies}
\begin{eqnarray*}
&& \hspace{-1.5em} \sum_{j=1}^N(\Lambda-\mu_j(A))_+^{\sigma+1/2}=\frac{\Gamma(\sigma+3/2)}
{\Gamma(\sigma-1/2)\Gamma(2)}\int_0^\infty(\Lambda-t)_+^{\sigma-3/2}\sum_{j=1}^N
(t-\mu_j(A))_+\,\mathrm{d}t \\ &&
\le\frac{\Gamma(\sigma+3/2)}{\Gamma(\sigma-
1/2)}\int_0^\infty(\Lambda-t)_+^{\sigma-3/2}\biggl(\frac{(t^2-B_0^2)_+
|\omega|}{8\pi} \\ && \quad + \frac{(t-B_0)_+B_0|\omega|}{4\pi}\left\{\frac
{\Lambda+B_0}{2B_0}\right\}\biggr)\,\mathrm{d}t \\ &&
\le\frac{\Gamma(\sigma+3/2)
|\omega|}{\Gamma(\sigma-1/2)}\biggl(\frac{(\Lambda^2-B_0^2)_+}{8\pi} \\ && \quad +
\frac{(\Lambda-B_0)_+B_0}{4\pi}\left\{\frac{\Lambda+B_0}
{2B_0}\right\}\biggr)\int_0^\infty(\Lambda-t)_+^{\sigma-3/2}\,\mathrm{d}t \\ && \
=\frac{\Gamma(\sigma+3/2)\Lambda^{\sigma-1/2}|\omega|}{\Gamma(\sigma-1/2)(2\sigma-1)}
\left(\frac{(\Lambda^2-B_0^2)_+}{4\pi}+\frac{(\Lambda-B_0)_+B_0}{2\pi}\left\{\frac{\Lambda+B_0}
{2B_0}\right\}\right).
\end{eqnarray*}
\end{remark}

\subsection{Comparison to earlier results}

Given a set $\omega\subset\mathbb{R}^2$ and a point $x\in\omega$, we denote by
$$
\delta(x)=\mathrm{dist}(x,\partial\omega)=\min_{y\in\partial\omega}|x-y|
$$
the distance of $x$ to the boundary, then
$$
R(\omega)=\sup_{x\in\omega}\delta(x)
$$
is the in-radius of $\omega$. Furthermore, given a $\beta>0$ we introduce
$$
\omega_\beta=\{x\in\omega:\,\,\delta(x)<\beta\}\,,\quad\beta>0\,,
$$
and define the quantity
 \begin{equation}\label{sigma}
\sigma(\omega):=\inf_{0<\beta<R(\omega)}\frac{|\omega_\beta|}
{\beta}\,.
\end{equation}
Using these notions and the symbols introduced above we can state the following result obtained in the work of two of us \cite{KW13}:

\begin{theorem} Let $\omega\subset\mathbb{R}^2$ be an open convex domain, then for any $\Lambda>B_0$ we have
\begin{eqnarray}\label{Ber.In}
\lefteqn{\sum_{j=1}^N(\Lambda-\mu_j(A))\le\frac{\Lambda^2|\omega|}{8\pi}
-\frac{1}{512\pi}\frac{\sigma^2(\omega)}{|\omega|}\Lambda} \\ &&
-B_0^2\left(\frac{1}{2}-\left\{\frac{\Lambda+B_0}{2
B_0}\right\}\right)^2\left(\frac{|\omega|}{2\pi}
-\frac{1}{128\pi}\frac{\sigma^2(\omega)}{|\omega|\Lambda}\right).
\nonumber
\end{eqnarray}
\end{theorem}

\bigskip

\noindent To make a comparison to the conclusions of the previous section, let us make both $B_0$ and $\Lambda$ large keeping their ratio fixed. Specifically, we choose a $\Lambda$ from the interval $(B_0,\,2B_0)$ writing it as $\Lambda=B_0(1+\alpha)$ with an $\alpha\in (0,1)$. The second term on the right-hand side of (\ref{Berezin}) is then $\frac{\alpha^2 B_0^2|\omega|}{8\pi}$, and we want to show that the difference between the bounds (\ref{Ber.In}) and (\ref{Berezin}) tends to plus infinity as $B_0\to\infty$. To this aim, we write $\Lambda=B_0(1+\alpha)$ with an $\alpha\in (0,1)$, then
 \begin{equation}\label{comparison}
\frac{(\Lambda^2-B_0^2)|\omega|}{8\pi}+
\frac{(\Lambda-B_0)B_0|\omega|}{4\pi}\left\{\frac{\Lambda+B_0}{2
B_0}\right\}= \frac{B_0^2|\omega|}{4\pi}\,\alpha(1+\alpha)\,.
\end{equation}
On the other hand, a short calculation shows that for our choice of $B_0$ and $\Lambda$ the right-hand side of the bound (\ref{Ber.In}) becomes
$$
=\frac{\Lambda^2|\omega|}{8\pi}-\frac{B_0^2|\omega|}{2\pi}\left(\frac{1}{2}-
\frac{\alpha}{2}\right)^2+\frac{\Lambda}{512\pi}\frac{\sigma^2(\omega)}{|\omega|}
\left(-1+\frac{(1-\alpha)^2}{(1+\alpha)^2}\right),
$$
in particular, after another easy manipulation we find that for large $B_0$ this expression behaves as $\frac{B_0^2|\omega|}{2\pi}\alpha +\mathcal{O}(B_0)$. Comparing the two bounds we see that
\begin{equation}\label{comp-diff}
\text{\emph{rhs of} (\ref{Ber.In})} - \text{\emph{rhs of} (\ref{Berezin})}=\frac{B_0^2|\omega|}{4\pi}\,\alpha(1-\alpha)+\mathcal{O}(B_0)
\end{equation}
tending to plus infinity as $B_0\to\infty$. At the same time,
\begin{equation}\label{comp-rat}
\frac{\text{\emph{rhs of} (\ref{Ber.In})}}{\text{\emph{rhs of} (\ref{Berezin})}} =\frac{2}{1+\alpha}\,+\mathcal{O}(B_0^{-1})
\end{equation}
illustrating that the improvement represented by Theorem~\ref{th-berezin} is most pronounced for eigenvalues near the spectral threshold.

\section{Examples: a two-dimensional disc}
\label{s: disc}
\setcounter{equation}{0}

Spectral analysis simplifies if the domain $\omega$ allows for a separation of variables. In this section we will discuss two such situations.

\subsection{Constant magnetic field}

We suppose that $\omega$ is a disc and the applied magnetic field is homogeneous. As usual in cases of a radial symmetry, the problem can be reduced to degenerate hypergeometric functions. Specifically, we will employ the Kummer equation
\begin{equation}\label{Kummer}
r\frac{\mathrm{d}^2\omega}{\mathrm{d}r^2}+(b-r)\frac{\mathrm{d}\omega}{\mathrm{d}r}-a\omega=0
\end{equation}
with real valued parameters $a$ and $b$ which has two linearly independent solutions  $M(a, b, r)$ and $U(a, b, r)$, the second one of which has a singularity at zero \cite{AS64}.

Given an $\alpha>0$, we denote by $\big\{a^k_{|m|, \alpha}\big\}_{k\in\mathbb{N}}$ the set of the first parameter values such that $M(a^k_{|m|, \alpha}, |m|+1, \alpha)=0$. Since for any $a, b\ge0$ the function $M(a, b, r)$ has no positive zeros \cite{AS64}, all the $a^k_{|m|, \alpha}$ are negative. Then the following claim is valid.

\begin{theorem} \label{thm:homdisc}
Let $H_\omega(A)$ be the magnetic Dirichlet Laplacian corresponding to a constant magnetic field $B_0$ and $\omega$ being the two dimensional disc with center at the origin and radius $r_0>0$. The eigenvalues of $H_\omega(A)$ coincides with
$$
\left\{B_0+B_0\left(|m|-m-2a^k_{|m|, \sqrt{B_0} \, r_0/ \sqrt{2}}\right)\right\}_{m\in\mathbb{Z},\, k\in\mathbb{N}}\,.
$$
\end{theorem}
\begin{proof}
We employ the standard partial wave decomposition -- see, e.g., \cite{Er96}
\begin{equation}\label{decomp.}
L^2(\omega)=\bigoplus_{m=-\infty}^\infty L^2((0, r_0), 2\pi r\,\mathrm{d}r)\,,
\end{equation}
and $H_\omega(A)=\bigoplus_{m=-\infty}^\infty h_m$, where
\begin{equation}
\label{hm*}h_m:=-\frac{\mathrm{d}^2}{\mathrm{d}r^2}-\frac{1}{r}\frac{\mathrm{d}}{\mathrm{d}r}+\left(\frac{m}{r}-\frac{B_0 r}{2}\right)^2.
\end{equation}
The last named operator differs by $mB_0$ from the operator
\begin{equation}\label{tildehm}
\tilde{h}_m=-\frac{\mathrm{d}^2} {\mathrm{d}r^2}-\frac{1}{r}\frac{\mathrm{d}}{\mathrm{d}r}+\frac{m^2}{r^2}+\frac{B_0^2 r^2}{4}
\end{equation}
on the interval $(0,r_0)$ with Dirichlet boundary condition at the endpoint $r_0$. Looking for solutions to the eigenvalue equation
\begin{equation}\label{magn.lapl.disc}
\tilde{h}_mu=\lambda u
\end{equation}
we employ the Ansatz
$$
u(r)=r^{|m|} e^{-B_0 r^2/4}v(r)\,,
$$
where $v\in L^2((0, r_0), r\mathrm{d}r)$.
Computing the first two derivatives we get
$$
\tilde{h}_mu=\left(-v''-\frac{2|m|+1}{r}v^{\prime}+B_0(|m|+1)v(r)+B_0 r v^{\prime}\right)r^{|m|} e^{-B_0 r^2/4},
$$
hence the equation (\ref{magn.lapl.disc}) can rewritten as
\begin{equation}\label{change.of.var.}
v''+\left(\frac{2|m|+1}{r}-B_0 r\right)v^{\prime}-(B_0(|m|+1)-\lambda)v=0\,.
\end{equation}
Using the standard substitution we pass to the function $g(r)=v\left(\frac{\sqrt{2r}}{\sqrt{B_0}}\right)$ belonging to $L^2\left(0, B_0 r_0^2/2\right)$.
Expressing the derivatives of $v$ in terms of those of $g$, one can rewrite equation (\ref{change.of.var.}) as
$$
rg''(r)+\left(|m|+1-r\right)g^{\prime}-\frac{\left((|m|+1)B_0-\lambda\right)}{2B_0}g(r)=0\,,
$$
which is the Kummer equation with $b=|m|+1$ and $a=\frac{(|m|+1)B_0-\lambda}{2B_0}$. The mentioned singularity of its solution $U(a, b, r)$ for small $r$, namely \cite{AS64}
$$
U(a, b, r)=\frac{\Gamma(b-1)}{\Gamma(a)}r^{1-b}+\mathcal{O}(r^{b-2})\quad\text{for}\quad b>2
$$
and
$$
U(a, 2, z)=\frac{1}{\Gamma(a)}\frac{1}{r}+\mathcal{O}(\ln r)\,,\quad U(a, 1, r)=-\frac{1}{\Gamma(a)}\ln r+\mathcal{O}(1)
$$
means that $u(r)=r^{|m|} e^{-B_0 r^2/4}U\left(\frac{(|m|+1)B_0-\lambda}{2B_0}, |m|+1, \frac{B_0\,r^2}{2}\right)$ does not belong to $\mathcal{H}_0^1((0, r_0),\,r\mathrm{d}r)$. Consequently, the sought solution of (\ref{magn.lapl.disc}) has the form
$$
r^{|m|} e^{-B_0 r^2/4}M\left(\frac{(|m|+1)B_0-\lambda}{2B_0}, |m|+1, \frac{B_0\,r^2}{2}\right)\,,
$$
and in view of the Dirichlet boundary conditions at $r_0$ we arrive at the spectral condition
$$
M\left(\frac{(|m|+1)B_0-\lambda}{2B_0}, |m|+1, \frac{B_0\,r_0^2}{2}\right)=0\,.
$$
which gives $\left\{(|m|+1)B_0-2B_0 a^k_{|m|, \sqrt{B_0}\,r_0/\sqrt{2}}\right\}_{m\in\mathbb{Z}, \,k\in\mathbb{N}}$ as the eigenvalue set; returning to the original operator $h_m$ we get the claim of the theorem.
\end{proof}

\subsection{Radial magnetic field}

If the magnetic field is non-constant but still radially symmetric, in general one cannot find the eigenvalues explicitly but it possible to find a bound to the eigenvalue moments in terms of an appropriate radial two-dimensional Schr\"odinger operator.

\begin{theorem} \label{thm:raddisc}
Let $H_\omega(A)$ be the magnetic Dirichlet Laplacian $H_\omega(A)$ on a disc $\omega$ of radius $r_0>0$ centered at the origin with a radial magnetic field $B(x)=B(|x|)$. Assume that
\begin{equation}\label{assump.}
\alpha:=\int_0^{r_0} sB(s)\,\mathrm{d}s<\frac{1}{2}\,.
\end{equation}
Then for any $\Lambda,\,\sigma\ge0$, the following inequality holds true
\begin{eqnarray}\label{magn.estimate}
\lefteqn{\mathrm{tr}(\Lambda-H_\omega(A))_+^\sigma\le\left(\frac{1}{\sqrt{1-2\alpha}}+\sup_{n\in\mathbb{N}}\left\{\frac{n}{\sqrt{1-2\alpha}}\right\}\right)} \\ &&
\times\:\mathrm{tr}\left(\Lambda-\left(-\Delta_D^\omega+\frac{1}{x^2+y^2}\left(\int_0^{\sqrt{x^2+y^2}} sB(s)\,\mathrm{d}s\right)^2\right)\right)_+^\sigma\,. \nonumber
\end{eqnarray}
In particular, the estimate (\ref{magn.estimate}) implies
$$
\inf \sigma(H_\omega(A))\ge\inf \sigma\left(-\Delta_D^\omega+\frac{1}{x^2+y^2}\left(\int_0^{\sqrt{x^2+y^2}} sB(s)\,\mathrm{d}s\right)^2\right)\,.
$$
\end{theorem}
\begin{proof}
Let us again employ the partial-wave decomposition (\ref{decomp.}), with the angular component (\ref{hm*}) replaced by
\begin{equation}\label{hm}
h_m:=-\frac{\mathrm{d}^2}{\mathrm{d}r^2}
-\frac{1}{r}\frac{\mathrm{d}}{\mathrm{d}r}
+\left(\frac{m}{r}-\frac{1}{r}\int_0^r sB(s)\,\mathrm{d}s\right)^2,
\end{equation}
and inspect the eigenvalues of this operator. Obviously, for $m\le0$ we have
\begin{equation}\label{h1m}
h_m\ge-\frac{\mathrm{d}^2}{\mathrm{d}r^2}
-\frac{1}{r}\frac{\mathrm{d}}{\mathrm{d}r}
+\frac{m^2}{r^2}+\frac{1}{r^2}\left(\int_0^r sB(s)\,\mathrm{d}s\right)^2,
\end{equation}
while for any $m>0$ we can use the inequality
$$
\frac{2|m|}{r^2}\int_0^r sB(s)\,\mathrm{d}s\le\frac{2m^2}{r^2}\int_0^r sB(s)\,\mathrm{d}s
$$
which in view of the assumption (\ref{assump.}) yields
$$
h_m\ge-\frac{\mathrm{d}^2}{\mathrm{d}r^2}
-\frac{1}{r}\frac{\mathrm{d}}{\mathrm{d}r}
+(1-2\alpha)\frac{m^2}{r^2}+\frac{1}{r^2}\left(\int_0^r sB(s)\,\mathrm{d}s\right)^2.
$$
Next we divide the set of natural numbers into groups such that for all the elements of any fixed group the entire part $\left[\sqrt{1-2\alpha}\,m\right]$ is the same, and we estimate the operator $h_m$ from below by
\begin{equation}\label{tildehm2}
h_m \ge -\frac{\mathrm{d}^2}{\mathrm{d}r^2}
-\frac{1}{r}\frac{\mathrm{d}}{\mathrm{d}r}
+\frac{\left[\sqrt{1-2\alpha}\, m\right]^2}{r^2}+\frac{1}{r^2}\left(\int_0^r sB(s)\,\mathrm{d}s\right)^2.
\end{equation}
Since the number of elements in each group is bounded from above by the sum $\frac{1}{\sqrt{1-2\alpha}} +\sup_{n\in\mathbb{N}}\left\{ \frac{n}{\sqrt{1-2\alpha}}\right\}$, using (\ref{h1m}) and (\ref{tildehm2}) one infers that
\begin{eqnarray*}
\lefteqn{\mathrm{tr}(\Lambda-H_\omega(A))_+^\sigma\le\left(\frac{1}{\sqrt{1-2\alpha}}
+\sup_{n\in\mathbb{N}}\left\{\frac{n}{\sqrt{1-2\alpha}}\right\}\right)} \\ &&
\times\sum_{m=-\infty}^\infty\mathrm{tr}\left(\Lambda-\left(-\frac{\mathrm{d}^2}{\mathrm{d}r^2}
-\frac{1}{r}\frac{\mathrm{d}}{\mathrm{d}r}+\frac{m^2}{r^2}+\frac{1}{r^2}\left(\int_0^r sB(s)\,\mathrm{d}s\right)^2\right)\right)_+^\sigma \\ &&
\hspace{-1em} =\left(\frac{1}{\sqrt{1-2\alpha}}+\sup_{n\in\mathbb{N}}\left\{\frac{n}{\sqrt{1-2\alpha}}\right\}\right)
\\ &&
\times\:\mathrm{tr}\left(\Lambda-\bigoplus_{m=-\infty}^\infty \left(-\frac{\mathrm{d}^2}{\mathrm{d}r^2}
-\frac{1}{r}\frac{\mathrm{d}}{\mathrm{d}r}+\frac{m^2}{r^2}+\frac{1}{r^2}\left(\int_0^r sB(s)\,\mathrm{d}s\right)^2\right)\right)_+^\sigma
\end{eqnarray*}
with any $\sigma,\,\Lambda\ge0$. However, the direct sum in the last expression is nothing else than a partial-wave decomposition of the two-dimensional Schr\"odinger operator with the radial potential $V(r) = \frac{1}{r^2}\left(\int_0^r sB(s)\,\mathrm{d}s\right)^2$ and the Dirichlet condition at the boundary of the disc; this yields the desired claim.
\end{proof}

\section{Application to the three-dimensional case}
\label{3Dapplication}
\setcounter{equation}{0}

Let us return now to our original motivation of estimating eigenvalues due to confinement in a three-dimensional `bottle'. One can employ inequality (\ref{magn.field}) in combination with the results of the previous sections to improve in some cases the spectral bound by taking the magnetic field into account instead of just dropping it.

Let $\Omega\subset\mathbb{R}^3$ with the bounded $x_3$ cross sections. The class of fields to consider are those of the form $B(x)=(B_1(x), B_2(x), B_3(x_3))$, that is, those for which the component $B_3$ perpendicular to the cross section depends on the variable $x_3$ only. Such fields certainly exist, for instance, one can think of the situation when the `bottle' is placed into a homogeneous magnetic field. The field is induced by an appropriate vector potential $A(\cdot):\Omega\to \mathbb{R}^3$,
$$
B(x)=(B_1(x),B_2(x),B_3(x_3))=\mathrm{rot}\,A(x),
$$
and we consider the magnetic Dirichlet Laplacians
$$
\mathcal{H}_\Omega(A)=(i\nabla_x-A(x))^2\quad\text{on}\;\, L^2(\Omega).
$$
We use the notion introduced in Sec.~\ref{s: reduction}. In view of the variational principle we know that the ground-state eigenvalue of $\widetilde{H}_{\omega(x_3)}(\widetilde{A})$ cannot fall below the first Landau level $B_3(x_3)$. Consequently, integrating with respect to $x_3$ in the formula (\ref{magn.field}) one can drop for all the $x_3$ for which $B_3(x_3)\ge\Lambda$. Combining this observation with Remark~\ref{AL-rmk} we get
\begin{eqnarray*}
\lefteqn{\mathrm{tr} (\Lambda-\mathcal{H}_\Omega(A))_+^\sigma\le \frac{\Gamma(\sigma+3/2)\Lambda^{\sigma-1/2}}{4\pi(2\sigma-1)
\Gamma(\sigma-1/2)}\,L_{1,\sigma}^{\mathrm{cl}}
\int_{\{x_3:\, B_3(x_3)<\Lambda\}}|\omega(x_3)|} \\[.5em] && \times\,\biggl[\big(\Lambda^2-B_3(x_3)^2\big)
+2B_3\big(\Lambda-B_3(x_3)\big)\,\left\{\frac{\Lambda+B_3}{2B_3}
\right\}\biggr]\,\mathrm{d}x_3
\end{eqnarray*}
for any $\sigma\ge3/2$.

\begin{example}
(circular cross section)
{\rm Let $\Omega$ be a three-dimensional cusp with a circular cross section $\omega(x_3)$ of radius $r(x_3)$ such that $r(x_3)\to0$ as $x_3\to\infty$.  Then the above formula in combination with Theorem~\ref{thm:homdisc} yields
\begin{eqnarray*}
\lefteqn{\mathrm{tr} (\Lambda-\mathcal{H}_\Omega(A))_+^\sigma\le L_{1,\sigma}^{\mathrm{cl}}\sum_{m\in\mathbb{Z}, \,k\in\mathbb{N}}\int_{\mathbb{R}}\biggl(\Lambda-B_3(x_3)} \\[.2em] && -B_3(x_3)\biggl(|m|-m-2a_{|m|, \sqrt{B_3(x_3)} r_0(x_3)/\sqrt{2}}^k\biggr)\biggr)_+^{\sigma+1/2}\,\mathrm{d}x_3
\end{eqnarray*}
for any $\sigma\ge3/2$. The particular case $B(x)=\{0,0,B\}$ applies to a cusp-shaped region placed to a homogeneous field parallel to the cusp axis.}
\end{example}
\begin{example}
(radial magnetic field)
{\rm Consider the same cusp-shaped region $\Omega$ in the more general situation when the third field component can depend on the radial variable, $B(x)=(B_1(x), B_2(x), B_3(x_1^2+x_2^2, x_3))$, assuming that
$$
\sup_{x_3\in\mathbb{R}} \alpha(x_3)=\sup_{x_3\in\mathbb{R}} \int_0^{r_0(x_3)} sB_3(s, x_3)\,\mathrm{d}s<\frac{1}{2}\,.
$$
Then the dimensional reduction in view of Theorem~\ref{thm:raddisc} gives
\begin{eqnarray*}
\lefteqn{\mathrm{tr} (\Lambda-\mathcal{H}_\Omega(A))_+^\sigma\le L_{1,\sigma}^{\mathrm{cl}}\int_{\mathbb{R}}
\left(\frac{1}{\sqrt{1-2\alpha(x_3)}}
+\sup_{n\in\mathbb{N}}\left\{\frac{n}{\sqrt{1-2\alpha(x_3)}}\right\}\right)} \\ && \times\,\mathrm{tr}\left(\Lambda-\left(-\Delta_D^{\omega(x_3)}
+\frac{1}{x_1^2+x_2^2}\left(\int_0^{\sqrt{x_1^2+x_2^2}} sB_3(s, x_3)\,\mathrm{d}s\right)^2\right)\right)_+^{\sigma+1/2}
\end{eqnarray*}
for any $\sigma\ge3/2$.}
\end{example}

\section{Spectral estimates for eigenvalues from \\ perturbed magnetic field}
\label{s:mgBerezin}
\setcounter{equation}{0}

Now we change the topic and consider situations when the discrete spectrum comes from the magnetic field alone. We are going to demonstrate a Berezin-type estimate for the magnetic Laplacian on $\mathbb{R}^2$ with the field which is a radial and local perturbation of a homogeneous one. We consider the operator $H(B)$ in $L^2(\mathbb{R}^2)$ defined as follows,
\begin{equation} \label{eq-HB}
H(B)= -\partial_x^2 +(i\partial_y+A_2)^2, \qquad  A = \big(0,B_0\,x- f(x,y)\big)\,,
\end{equation}
with $f$ given by
$$ 
f(x, y)=-\int_x^\infty g(\sqrt{t^2+y^2})\,\mathrm{d}t\,.
$$ 
with $g: \mathbb{R}_+\to\mathbb{R}_+$; the operator $H(B)$ is then associated with the magnetic field
$$ 
B= B(x,y) = B_0 - g(\sqrt{x^2+y^2}\,)\,.
$$ 
Since have chosen the vector potential in such a way that the unperturbed  part corresponds to the Landau gauge, we have
$$ 
H(B_0)=-\partial_x^2 +(i\partial_y+B_0 x)^2.
$$ 
Using a partial Fourier transformation, it is easy to conclude from here that the corresponding spectrum consists of identically spaced eigenvalues of infinite multiplicity, the Landau levels,
\begin{equation} \label{sp-hb0}
\sigma (H(B_0)) =\left\{(2n-1)B_0, \ \ n\in\mathbb{N}\, \right\} \, .
\end{equation}
It is well known that $\inf\sigma_{\mathrm{ess}}(H(B)-B)=0$, hence the relative compactness of $B_0-B$ with respect to $H(B)-B_0$ in $L^2(\mathbb{R}^2)$ implies
$$ 
\inf\sigma_{\mathrm{ess}}(H(B)) = B_0.
$$ 
We have to specify the sense in which the magnetic perturbation is local. In the following we will suppose that
\begin{enumerate}[(i)]
\item the function $g\in L^\infty(\mathbb{R}_+)$ is non-negative and such that both $f$ and $\partial_{x_2} f$ belong to $L^\infty(\mathbb{R}^2)$, and
$$
\lim_{x_1^2+x_2^2\to\infty} \big(|\partial_{x_2} f(x_1, x_2)| + |f(x_1, x_2)|\big) =0\,.
$$ 
\item $\|g\|_\infty \leq B_0\,.$
\end{enumerate}
Let us next rewrite the vector potentials $A_0$ and $A$ associated to $B_0$ and $B$ in the polar coordinates. Passing to the circular gauge we obtain
\begin{equation} \label{poincare}
A_0=(0, a_0(r))\,, \qquad   A=(0, a(r))\,,
\end{equation}
with
\begin{equation}\label{a-a0}
a_0(r) = \frac{B_0 r}{2}\,, \qquad a(r) = \frac{B_0 r}{2} - \frac{1}{r} \int_0^r g(s)\, s\, \mathrm{d}s\,.
\end{equation}
Hence the operators $H(B_0)$ and $H(B)$ are associated with the closures of the quadratic forms in $L^2(\mathbb{R}_+, r \mathrm{d}r)$ with the values
\begin{equation} \label{q-b-0}
Q(B_0)[u] = \int_0^{2\pi} \int_0^\infty \left ( |\partial_r u|^2 + |\, i  r^{-1} \partial_\theta u+a_0(r)\, u|^2\right)\, r\, \mathrm{d}r\, \mathrm{d}\theta
\end{equation}
and
\begin{equation} \label{q-b}
Q(B)[u] = \int_0^{2\pi} \int_0^\infty \left ( |\partial_r u|^2 + |\, i  r^{-1} \partial_\theta u+a(r)\, u|^2\right)\, r\, \mathrm{d}r\,\mathrm{d}\theta\,,
\end{equation}
respectively, both defined on $C_0^\infty(\mathbb{R}_+)$. Furthermore,
for every $k\in\mathbb{N}_0$ we introduce the following auxiliary potential,
\begin{equation} \label{vk}
 \qquad V_k(r) := \frac{2k}{r} (a_0(r)-a(r)) +a^2(r)-a_0^2(r)\,,
\end{equation}
and the functions
\begin{equation} \label{psi-k}
\psi_k(r) = \sqrt{\frac{B_0}{\Gamma(k+1)}}\ \left(\frac{B_0}{2}\right)^{k/2}\, r^k\, \exp \left( -\frac{B_0\, r^2}{4}\right).
\end{equation}
Finally let us denote by
\begin{equation} \label{flux.}
\alpha =  \int_0^\infty g(r)\, r\, \mathrm{d}r
\end{equation}
the flux associated with the perturbation; recall that in the rational units we employ the flux quantum value is $2\pi$. Now we are ready to state the result.

\begin{theorem} \label{thm-red}
Let the assumptions (i) and (ii) be satisfied, and suppose moreover that $\alpha \leq 1$. Put
\begin{equation} \label{lam-k.}
\Lambda_k =    \left(\psi_k, \big(\,V_k(\,\cdot)\big)_- \, \psi_k\right)_{L^2(\mathbb{R}_+, r \mathrm{d}r)}\,.
\end{equation}
Then the inequality
\begin{equation} \label{eq-LT-gen-2d}
\mathrm{tr} (H(B)-B_0)_-^\gamma \ \leq  2^\gamma \sum_{k=0}^\infty\ \Lambda_k^\gamma\,,  \qquad \gamma \ge 0\,,
\end{equation}
holds true whenever the right-hand side is finite.
\end{theorem}

\begin{remark}
{\rm For a detailed discussion of the asymptotic distribution of eigenvalue of the operator $H(B)$ we refer to \cite{rt08}.}
\end{remark}

\begin{proof}
We are going to employ the fact that both $A_0$ and $A$ are radial functions, see \eqref{poincare}, and note that by the partial-wave decomposition
\begin{equation} \label{pwd-mg}
\mathrm{tr} \, (H(B)-B_0)_-^\gamma = \sum_{k\in\mathbb{Z}}\, \mathrm{tr} \, (h_k(B)-B_0)_-^\gamma\,,
\end{equation}
where the operators $h_k(B)$ in $L^2(\mathbb{R}_+, r \mathrm{d}r)$ are associated with the closures of the quadratic forms
$$
Q_k[u] = \int_0^\infty \left ( |\partial_r u|^2 + \biggl|\frac{k}{r} u-a(r)\, u\biggr|^2\right)\, r\, \mathrm{d}r \,,
$$
defined originally on $C_0^\infty(\mathbb{R}_+)$, and acting on their domain as
$$ 
h_k(B) = -\partial_r^2 -\frac 1r \partial_r + \left(\frac kr  -a(r)\right)^2.
$$ 
In view of \eqref{vk} it follows that
$$
h_k(B) = h_k(B_0) + V_k(r)\,,
$$
where
$$
h_k(B_0) = -\partial_r^2 -\frac 1r \partial_r + \left(\frac kr  -a_0(r)\right)^2 \,.
$$
To proceed we need to recall some spectral properties of the two-dimensional harmonic oscillator,
$$
{\rm H}_\mathrm{osc} = -\Delta +\frac{B_0^2}{4}\, (x^2+y^2) \qquad \text{in} \;\, L^2(\mathbb{R}^2)\,.
$$
It is well known that the spectrum of ${\rm H}_\mathrm{osc}$ consists of identically spaced eigenvalues of a finite multiplicity,
\begin{equation} \label{h-osc-spec}
\sigma\bigr({\rm H}_\mathrm{osc}\bigl) = \left\{\, nB_0, \ n\in\mathbb{N} \,\right\} \,,
\end{equation}
where the first eigenvalue $B_0$ is simple and has a radially symmetric eigenfunction. The latter corresponds to the term with $k=0$ in the partial-wave decomposition of ${\rm H}_\mathrm{osc}$, which implies
$$ 
\sigma\bigl({\rm H}_\mathrm{osc}\bigr)  = \bigcup_{k\in\mathbb{Z}}\  \sigma \left(-\partial_r^2 -\frac 1r \partial_r + \frac{k^2}{r^2} +\frac{B_0^2\, r^2}{4}\right)\,,
$$ 
where the operators in the brackets at the right-hand side act in $L^2(\mathbb{R}_+, r \mathrm{d}r)$. Hence in view of \eqref{h-osc-spec} we have
\begin{equation} \label{k-neq-0}
\inf_{k\neq 0} \sigma \left(-\partial_r^2 -\frac 1r \partial_r + \frac{k^2}{r^2} +\frac{B_0^2\, r^2}{4}\right) \, \geq 2 B_0\, .
\end{equation}
On the other hand, for $k<0$ it follows from (ii), \eqref{vk} and \eqref{flux.} that
$$
V_k(r) = \frac{2k}{r}\, \int_0^r g(s)\, s\, \mathrm{d}s -B_0\, \int_0^r g(s)\, s\, \mathrm{d}s + \frac{1}{r^2} \left(\int_0^r g(s)\, s\, \mathrm{d}s\right)^{\!2}  \\
\geq k B_0 - B_0\,.
$$
By \eqref{k-neq-0} we thus obtain the following inequality which holds in the sense of quadratic forms on $C_0^\infty(\mathbb{R}_+)$ for any $k<0$,
\begin{eqnarray*}
\lefteqn{h_k(B) = h_k(B_0) + V_k(r) = -\partial_r^2 -\frac 1r \partial_r + \frac{k^2}{r^2} +\frac{B_0^2\, r^2}{4} - k B_0 +V_k(r)} \\ &&
\geq  -\partial_r^2 -\frac 1r \partial_r + \frac{k^2}{r^2} +\frac{B_0^2\, r^2}{4} -\alpha\, B_0 \\ &&
\geq (2-\alpha) B_0\,. \phantom{AAAAAAAAAAAAAAAAAAAAAAAAAAAA}
\end{eqnarray*}
Since $\alpha \leq 1$ holds by hypothesis, this implies that
\begin{equation} \label{k-pos}
\mathrm{tr} \, (H(B)-B_0)_-^\gamma = \sum_{k\in\mathbb{Z}}\, \mathrm{tr} \, (h_k(B)-B_0)_-^\gamma\,  = \sum_{k\geq 0}\, \mathrm{tr} \, (h_k(B)-B_0)_-^\gamma\, ,
\end{equation}
see \eqref{pwd-mg}. In order to estimate $\mathrm{tr} \, (h_k(B)-B_0)_-^\gamma$ for $k\geq 0$ we employ
$$
\Pi_k = \left( \cdot\, , \, \psi_k\right)_{L^2(\mathbb{R}_+, r \mathrm{d}r)}\, \psi_k\,,
$$
the projection onto the subspace spanned by $\psi_k$, and note that
\begin{equation} \label{norm}
\psi_k \in \ker (h_k(B_0)-B_0)\,, \qquad \|\psi_k\|_{L^2(\mathbb{R}_+, r \mathrm{d}r)} = 1 \quad \forall\ k\in\mathbb{N}\cup\{0\} \,.
\end{equation}
Let $Q_k = 1-\Pi_k$. From the positivity of $\bigl(\,V_k(\cdot)\bigr)_-$ it follows  that for any $u\in C_0^\infty(\mathbb{R}_+)$  it holds
\begin{eqnarray}
\lefteqn{\left( u, \left( \Pi_k \bigl(\,V_k(\cdot)\bigr)_- Q_k +  Q_k \bigl(\,V_k(\cdot)\bigr)_-\Pi_k\right)\, u\right)} \nonumber \\ && \leq \, \left( u,  \Pi_k \bigl(\,V_k(\cdot )\bigr)_- \Pi_k\, u\right) + \left( u,  Q_k \big(\,V_k(\cdot )\big)_- Q_k\, u\right)\,, \label{pq}
\end{eqnarray}
where the scalar products are taken in $L^2(\mathbb{R}_+, r \mathrm{d}r)$. From \eqref{pq} we infer that
\begin{eqnarray}
\lefteqn{h_k(B)-B_0 =  (\Pi_k+Q_k) \left(h_k(B_0) -B_0+V_k(\cdot)\right) (\Pi_k+Q_k)}
\nonumber \\
&& \geq (\Pi_k+Q_k) \left(h_k(B_0) -B_0-\bigl(\,V_k(\cdot)\bigr)_-\right) (\Pi_k+Q_k) \nonumber \\
&& \geq   \Pi_k \left(h_k(B_0) -B_0-2 \bigl(\,V_k(\cdot )\bigr)_-\right) \Pi_k \nonumber\\
&& \quad + Q_k \left(h_k(B_0) -B_0-2 \bigl(\,V_k(\cdot )\bigr)_-\right) Q_k\,. \label{pk-only}
\end{eqnarray}
The operator $h_k(B_0)$ has for each $k\in\mathbb{N}_0$ discrete spectrum which consists of simple eigenvalues. Moreover, from the partial-wave decomposition of the operator $H(B_0)$ we obtain
$$
\sigma (H(B_0)) =\left\{(2n-1)B_0, \ n\in\mathbb{N} \right\} \, = \bigcup_{k\in\mathbb{Z}} \, \sigma(h_k(B_0))\,,
$$
see \eqref{sp-hb0}. It means that
$$
\forall\ k \in\mathbb{Z}\ : \quad \sigma(h_k(B_0)) \, \subset\, \left\{(2n-1)B_0, \ \ n\in\mathbb{N} \right\}\,,
$$
and since $\psi_k$ is an eigenfunction of $h_k(B_0)$ associated to the simple eigenvalue $B_0$, see \eqref{norm}, it follows that
\begin{equation}   \label{higher-ev}
Q_k \left(h_k(B_0) -B_0\right) Q_k \, \geq\, 2 B_0\, Q_k\,, \qquad \forall\ k\in\mathbb{N}\cup\{0\}\,.
\end{equation}
On the other hand, by \eqref{vk} and \eqref{flux.} we infer
$$
\sup_{r>0} \bigl(\,V_k(r)\bigr)_- \, \leq\, \alpha\, B_0 \qquad \forall\ k\in\mathbb{N}\cup\{0\}.
$$
The last two estimates thus imply that
$$
Q_k \left(h_k(B_0) -B_0-2 \bigl(\,V_k(\cdot )\bigr)_-\right) Q_k \geq
Q_k \left( 2\, B_0 (1-\alpha)\right) Q_k  \geq 0\,,
$$
where we have used the assumption $\alpha \leq 1$. With the help of \eqref{pk-only} and the variational principle we then conclude that
\begin{eqnarray*}
\lefteqn{\mathrm{tr} \, (h_k(B)-B_0)_-^\gamma\,\leq\, \mathrm{tr} \left(\Pi_k \left(h_k(B_0) -B_0-2 \bigl(\,V_k(\cdot )\bigr)_-\right) \Pi_k\right)_-^\gamma} \\
&& = \mathrm{tr} \left(-2\, \Pi_k \bigl(\,V_k(\cdot )\big)_- \Pi_k\right)_-^\gamma = 2^\gamma\, \mathrm{tr} \left( \Pi_k \bigl(\,V_k(\cdot )\bigr)_- \Pi_k\right)^\gamma \\
&& = 2^\gamma\, \left(\psi_k, \bigl(\,V_k(\cdot )\bigr)_- \, \psi_k\right)^\gamma_{L^2(\mathbb{R}_+, r \mathrm{d}r)}\,
= 2^\gamma\,  \Lambda_k^\gamma\,, \phantom{AAAAAA}
\end{eqnarray*}
see \eqref{lam-k.}. To complete the proof it now remains to apply equation \eqref{k-pos}.
\end{proof}

\section{Three dimensions: a magnetic `hole'}
\setcounter{equation}{0}
\label{s:3Dhole}

Let us return to the three-dimensional situation and consider a magnetic Hamiltonian $\mathcal{H}(B)$ in $L^2(\mathbb{R}^3)$ associated to the magnetic field \mbox{$B :\mathbb{R}^3\to \mathbb{R}^3$} regarded as a perturbation of a homogeneous magnetic field of intensity \mbox{$B_0>0$} pointing in the $x_3$-direction,
\begin{equation} \label{mgf-3d}
B(x_1, x_2, x_3)=  (0, 0, B_0) - b(x_1, x_2, x_3)\,,
\end{equation}
with the perturbation $b$ of the form
$$ 
b(x_1, x_2, x_3) = \biggl(-\omega'(x_3)\, f(x_1, x_2),\,  0,\,  \omega(x_3)\, g\left(\sqrt{x_1^2+x_2^2}\, \right) \biggr)\,.
$$ 
Here $\omega: \mathbb{R}\to \mathbb{R}_+$ , $g: \mathbb{R}_+\to \mathbb{R}_+$ and
\begin{equation} \label{fgh2}
f(x_1, x_2) = -\int_{x_1}^\infty g\left(\sqrt{t^2+x_2^2}\, \right)\,  \mathrm{d}t\,.
\end{equation}
The resulting field $B$ thus has the component in the $x_3$-direction given the $B_0$ plus a perturbation which is a radial field in the $x_1, x_2-$plane with a $x_3-$dependent amplitude $\omega(x_3)$.  The first component of $B$ then ensures that $\nabla\cdot B=0$, which is required by the Maxwell equations which include no magnetic monopoles; it vanishes if the field is $x_3$-independent.

A vector potential generating this field can be chosen in the form
$$ 
A(x_1, x_2, x_3) = (0,\, B_0\, x_1 -\omega(x_3)\, f(x_1, x_2),\, 0 )\,,
$$ 
which reduces to Landau gauge in the unperturbed case, and consequently, the operator $\mathcal{H}(B)$ acts on its domain as
\begin{equation} \label{eq-H}
\mathcal{H}(B) = -\partial_{x_1}^2 +(i\partial_{x_2} +B_0\, x_1 -\omega(x_3)\, f(x_1, x_2) )^2 -\partial_{x_3}^2.
\end{equation}
We have again to specify the local character of the perturbation: we will suppose that
\begin{enumerate}[(i)]
\item the function $g\in L^\infty(\mathbb{R}_+)$ is non-negative, such that $f$ and $\partial_{x_2} f$ belong to $L^\infty(\mathbb{R}^2)$, and
$$ 
\lim_{x_1^2+x_2^2\to\infty} \big(|\partial_{x_2} f(x_1, x_2)| + |f(x_1, x_2)|\big) =0\,,
$$ 
\item $\omega\geq 0$, $\, \omega\in L^2(\mathbb{R})\cap L^\infty(\mathbb{R})$, and
$$ 
\|\omega\|_\infty\, \|g\|_\infty \leq B_0\,, \qquad   \lim_{|x_3|\to\infty}  \omega(x_3)  =0\,.
$$ 
\end{enumerate}

\begin{lemma}\label{lem-es}
The assumptions (i) and (ii) imply $\sigma_{\mathrm{ess}}(\mathcal{H}(B)) = [B_0,\infty)$.
\end{lemma}
\begin{proof}
We will show that the essential spectrum of $\mathcal{H}(B)$ coincides with the essential spectrum of the operator
$$ 
\mathcal{H}(B_0) = -\partial_{x_1}^2 +(i\partial_{x_2} +B_0\, x_1)^2 -\partial_{x_3}^2\,,
$$ 
which is easy to be found, we have $\sigma(\mathcal{H}(B_0))= \sigma_{\mathrm{ess}}(\mathcal{H}(B_0)) = [B_0,\infty)$. Let
$$
T= \mathcal{H}(B) -\mathcal{H}(B_0) = -2\, \omega f\, (i\partial_{x_2} + B_0 x_1) -i \omega\, \partial_{x_2} f +\omega^2 f^2.
$$
From assumption (i) in combination with \cite[Thm.~5.7.1]{da} it follows that the operator $(\omega\, \partial_{x_2} f +\omega^2 f^2 )(-\Delta+1)^{-1}$ is compact on $L^2(\mathbb{R}^3)$. The diamagnetic inequality and \cite{pi} thus imply that the sum $i \omega\, \partial_{x_2} f +\omega^2 f^2$ is relatively compact with respect to $\mathcal{H}(B_0)$.

As for the first term of the perturbation $T$, we note that since $(i\partial_{x_2} + B_0 x_1)$ commutes with $\mathcal{H}(B_0)$, it holds
\begin{eqnarray}
\lefteqn{\omega f\, (i\partial_{x_2} + B_0 x_1)\, (\mathcal{H}(B_0)+1)^{-1}} \nonumber \\ && = \omega f\, (\mathcal{H}(B_0)+1)^{-1/2}\, (i\partial_{x_2} + B_0 x_1)\, (\mathcal{H}(B_0)+1)^{-1/2}. \label{T-1}
\end{eqnarray}
In the same way as above, with the help of \cite[Thm.5.7.1]{da}, diamagnetic inequality, and \cite{pi}, we conclude that $\omega\,  f\, (\mathcal{H}(B_0)+1)^{-1/2}$ is compact on $L^2(\mathbb{R}^3)$. On the other hand, $(i\partial_{x_2} + B_0 x_1)\, (\mathcal{H}(B_0)+1)^{-1/2}$ is bounded on $L^2(\mathbb{R}^3)$. As their product the operator \eqref{T-1} is compact; by Weyl's theorem we then have $\sigma_{\mathrm{ess}}(\mathcal{H}(B))=\sigma_{\mathrm{ess}}(\mathcal{H}(B_0))= [B_0,\infty)$.
\end{proof}

\subsection{\bf Lieb-Thirring-type inequalities for $\mathcal{H}(B)$}
\label{LT-3D}

Now we are going to formulate Lieb-Thirring-type inequalities for the negative eigenvalues of $\mathcal{H}(B)-B_0$ in three different cases corresponding to different types of decay conditions on the function $g$. Let us start from a general result. We denote by
$$ 
\alpha(x_3) = \omega(x_3) \int_0^\infty g(r)\, r\, \mathrm{d}r
$$ 
the magnetic flux (up to the sign) through the plane $\{ (x_1,x_2, x_3): (x_1,x_2)\in\mathbb{R}^2\}$ associated with the perturbation. From  Theorem~\ref{thm-red} and inequality (\ref{magn.field}) we make the following conclusion.

\begin{theorem} \label{thm-red3}
Let assumptions (i) and (ii) be satisfied. Suppose, moreover, that $\sup_{x_3} \alpha(x_3) \leq 1$ and put
\begin{equation} \label{lam-k}
\Lambda_k(x_3) = \bigl(\psi_k, \big(\,V_k(\cdot ; x_3)\big)_- \, \psi_k\bigr)_{L^2(\mathbb{R}_+, r \mathrm{d}r)}\,.
\end{equation}
Then the inequality
\begin{equation} \label{eq-LT-gen-3d}
\mathrm{tr} \, (\mathcal{H}(B)-B_0)_-^\sigma \leq  L^{\mathrm{cl}}_{\sigma, 1} \ 2^{\, \sigma+\frac 12} \int_{\mathbb{R}} \, \sum_{k=0}^\infty\ \Lambda_k(x_3)^{\sigma+\frac 12}\, \mathrm{d}x_3\,,  \quad \sigma \geq \frac 32\,,
\end{equation}
holds true whenever the right-hand side is finite.
\end{theorem}

\subsubsection{Perturbations with a power-like decay}

Now we come to the three cases mentioned above, stating first the results and then presenting the proofs. We start from magnetic fields (\ref{mgf-3d}) with the perturbation $g$ which decays in a powerlike way. Specifically, we shall assume that
\begin{equation} \label{g-power}
0 \, \leq \, g(r) \, \leq \, B_0\, (1+\sqrt{B_0}\ r)^{-2\beta}\,, \qquad \beta >1\,.
\end{equation}
We have included the factor $\sqrt{B_0}$ on the right hand side of \eqref{g-power} having in mind that $B_0^{-1/2}$ is the Landau magnetic length which defines a natural length unit in our model.

\smallskip

\noindent For any $\beta> 1$ and $\gamma> \max\left\{ \frac{1}{\beta-1}\,, 2\right\}$ we define the number
\begin{equation} \label{eq-K}
K(\beta, \gamma) = 2^{-\gamma} +\sum_{k=1}^\infty\, \left(\frac{\Gamma\left (( k+1-\beta)_+\right)}{\Gamma(k)} +\frac{1}{2\sqrt{2\pi k}} \right)^\gamma\,,
\end{equation}
and recall also the classical Lieb-Thirring constants in one dimension,
\begin{equation} \label{LT-constants}
L^{\mathrm{cl}}_{1,\sigma}= \frac{\Gamma(\sigma+1)}{2\sqrt{\pi}\ \Gamma(\sigma+3/2)}\,, \qquad \sigma>0\,.
\end{equation}

\smallskip

\begin{theorem} \label{thm-power}
Assume that $g$ satisfies \eqref{g-power} and that $\|\omega\|_\infty \leq 2(\beta-1)$. Then
$$ 
\mathrm{tr} \, (\mathcal{H}(B)-B_0)_-^\sigma \ \leq \ L^{\mathrm{cl}}_{1,\sigma}\  K\Big(\beta, \sigma+\frac 12\Big) \left(\frac{2\, B_0}{\beta-1}\right)^{\sigma+\frac12} \int_{\mathbb{R}}\omega(x_3)^{\sigma+\frac 12}\, \mathrm{d}x_3
$$ 
holds true for all
\begin{equation} \label{sigma-min}
\sigma > \max\left\{ \frac 32\, ,\,  \frac{3-\beta}{2\beta-2} \right\}\,.
\end{equation}
\end{theorem}

\begin{remark}
{\rm Since $\omega\in L^\infty(\mathbb{R}) \cap L^2(\mathbb{R})$, it follows that $\omega\in L^{\sigma+\frac 12}(\mathbb{R})$ for any $\sigma\geq 3/2$. Note also that by the Stirling formula we have}
$$
\frac{\Gamma\left ( k+1-\beta \right)}{\Gamma(k)} \ \sim \ k^{1-\beta} \quad \mathrm{as}\quad k\to\infty\,.
$$
{\rm Hence the constant $K\big(\beta, \sigma+\frac 12\big)$ is finite for any $\sigma$ satisfying \eqref{sigma-min}.}
\end{remark}

\subsubsection{Gaussian decay}

Next we assume that the perturbation $g$ has a Gaussian decay, in other words
\begin{equation} \label{g-gauss}
0 \, \leq \, g(r) \, \leq \, B_0\,  e^{- \varepsilon B_0 r^2} \, , \qquad \varepsilon >0.
\end{equation}

\smallskip

\begin{theorem} \label{thm-gauss}
Assume that $g$ satisfies \eqref{g-gauss} and that $\|\omega\|_\infty \leq 2 \varepsilon$. Then for any $\sigma > 3/2$ it holds
$$ 
\mathrm{tr} \, (\mathcal{H}(B)-B_0)_-^\sigma \ \leq \  L^{\mathrm{cl}}_{\sigma,1}\,  \left(\frac{B_0}{\varepsilon}\right)^{\sigma+\frac 12} G(\varepsilon, \sigma)\,  \int_{\mathbb{R}} \omega(x_3)^{\sigma+\frac 12}\, \mathrm{d}x_3\,,
$$ 
where
\begin{equation}\label{G}
G(\varepsilon, \sigma) = 1 + \sum_{k=1}^\infty \left((1+2\varepsilon)^{-k} + \frac{1}{2\sqrt{2\pi k}}\right)^{\sigma+\frac 12}\,.
\end{equation}\end{theorem}

\subsubsection{Perturbations with a compact support}

Let $D$ be a circle of radius $R$ centered at the origin and put
\begin{equation} \label{g-hole}
g(r) = \left\{
\begin{array}{l@{\quad}l}
B_0                   & \quad r \leq R \, \\
0 &\quad r > R  \\
\end{array}
\right. \, .
\end{equation}

\smallskip

\begin{theorem} \label{thm-hole}
Assume that $g$ satisfies \eqref{g-hole} with $R$ such that $B_0 R^2\leq 2$. Suppose moreover that $\|\omega\|_\infty \leq 1$. Then for any $\sigma > 3/2$ it holds
\begin{equation} \label{LT-hole-3d}
\mathrm{tr} \, (\mathcal{H}(B)-B_0)_-^\sigma \ \leq \  L^{\mathrm{cl}}_{\sigma,1}\ J\Big(B_0\, , \sigma\Big)\ B_0^{\sigma+\frac 12}
\int_{\mathbb{R}} \omega(x_3)^{\sigma+\frac 12}\, \mathrm{d}x_3\,,
\end{equation}
where
\begin{equation}
J(B_0, \sigma) = \left(B_0\, R^2\right)^{\sigma+\frac 12}\left( 1 + \sum_{k=1}^\infty  \left( \left(\frac{B_0\, R^2}{2}\right)^{k+1} \frac{1}{k!} + \frac{1}{2\sqrt{2\pi k}}\right)^{\sigma+\frac 12}\, \right)\, .
\end{equation}
\end{theorem}

\subsection{The proofs}

Note that the assumptions of these theorems ensure that $\sup_{x_3} \alpha(x_3)\leq 1$, hence in all the three cases we may apply Theorem~\ref{thm-red} and, in particular, the estimate \eqref{eq-LT-gen-3d}. To this note it is useful to realize that by \eqref{a-a0}, \eqref{vk} and \eqref{flux.} we have
\begin{eqnarray}
\lefteqn{V_k(r;x_3) = -\alpha(x_3) B_0 +\frac{2 \alpha(x_3) k}{r^2} - \frac{2k\, \omega(x_3)}{r^2} \int_r^\infty g(s)\, s\, \mathrm{d}s} \nonumber \\
&& \quad +B_0\, \omega(x_3) \int_r^\infty g(s)\, s\, \mathrm{d}s + \frac{\omega^2(x_3)}{r^2} \left(\int_0^r g(s)\, s\, \mathrm{d}s\right)^2. \label{vk-eq}
\end{eqnarray}
Consequently, we obtain a simple upper bound on the negative part of $V_k$,
\begin{equation} \label{vk-neg}
\big(\,V_k(r ; x_3)\big)_- \ \leq\ \frac{2k\, \omega(x_3)}{r^2} \int_r^\infty g(s)\, s\, \mathrm{d}s +\alpha(x_3)\left(B_0 -\frac{2k}{r^2}\right)_+
\end{equation}
for all $k\in\mathbb{N}\cup\{0\}$. For $k=0$ we clearly we have
\begin{equation} \label{upperb-0}
\Lambda_0(x_3) \, \leq\,  \alpha(x_3) B_0\,,
\end{equation}
by \eqref{norm}. In order to estimate $\Lambda_k(x_3)$ with $k\geq 1$ we denote by $\lambda_k(x_3)$ the contribution to $\Lambda_k(x_3)$ coming from the first term on the right-hand side of \eqref{vk-neg}, i.e.
\begin{equation} \label{lam-k-2}
\lambda_k(x_3) =  2\, \omega(x_3)\, k \int_0^\infty \psi^2_k(r) \left(\int_r^\infty g(s)\, s\, \mathrm{d}s\right) \, r^{-1}\, \mathrm{d}r\,.
\end{equation}
Before coming to the proofs we need an auxiliary result.

\begin{lemma} \label{lem-aux}
For any $k\in\mathbb{N}$ it holds
$$
\Lambda_k(x_3) \, \leq\, \lambda_k(x_3) + \frac{\alpha(x_3) B_0}{\sqrt{2\pi k}}\,.
$$
\end{lemma}
\begin{proof}
In view of \eqref{lam-k}, \eqref{vk-neg}, and \eqref{lam-k-2} the claim will follow if we show that
\begin{equation} \label{enough}
 \int_0^\infty \psi^2_k(r) \, \left(B_0 -\frac{2k}{r^2}\right)_+ r\, \mathrm{d}r\ \leq \ \frac{B_0}{\sqrt{2\pi k}}\,.
\end{equation}
Let $r_k = \sqrt{\frac{2k}{B_0}}$. Using \eqref{psi-k} and the substitution $s= \frac{B_0 r^2}{2}$ we then find
\begin{eqnarray*}
\lefteqn{\int_0^\infty \psi^2_k(r) \, \left(B_0 -\frac{2k}{r^2}\right)_+ r\, \mathrm{d}r = B_0  \int_{r_k}^\infty \psi^2_k(r) \, r\, \mathrm{d}r - 2k
  \int_{r_k}^\infty \psi^2_k(r) \, r^{-1}\, \mathrm{d}r} \\
  && = \frac{B_0}{\Gamma(k+1)} \int_k^\infty e^{-s}\, s^k\, \mathrm{d}s -\frac{B_0}{\Gamma(k)}  \int_k^\infty e^{-s}\, s^{k-1}\, \mathrm{d}s\,. \phantom{AAAAAAAAAAA}
\end{eqnarray*}
Integration by parts gives
$$
\int_k^\infty e^{-s}\, s^k\, \mathrm{d}s = e^{-k}\, k^k + k \int_k^\infty e^{-s}\, s^{k-1}\, \mathrm{d}s\,,
$$
hence
$$
 \int_0^\infty \psi^2_k(r) \, \left(B_0 -\frac{2k}{r^2}\right)_+ r\, \mathrm{d}r = \frac{e^{-k}\, k^k\, B_0}{\Gamma(k+1)}\,,
$$
and inequality \eqref{enough} follows from the Stirling-type estimate \cite[Eq.~6.1.38]{AS64}
$$
\Gamma(k+1) = k! \geq \sqrt{2\pi}\ k^{\, k+\frac 12}\, e^{-k}\,, \qquad k\in\mathbb{N}\,;
$$
this concludes the proof.
\end{proof}

\smallskip

\begin{proof}[\bf Proof of Theorem \ref{thm-power}]
In view of \eqref{upperb-0} and Lemma \ref{lem-aux} it suffices to estimate $\lambda_k(x_3)$ in a suitable way from above for $k\geq 1$. Using \eqref{g-power} we find
\begin{align*}
\int_0^\infty g(r)\, r \, \mathrm{d}r & \leq B_0 \int_0^\infty
(1+\sqrt{B_0}\ r)^{-2\beta}\, r\,  \mathrm{d}r \leq B_0 \int_0^\infty (1+\sqrt{B_0}\ r)^{1-2\beta}\, \mathrm{d}r \\
&=  \int_0^\infty (1+ s)^{1-2\beta}\, \mathrm{d}s = \frac{1}{2(\beta-1)}\,,
\end{align*}
which implies
\begin{equation} \label{alpha-upperb}
\alpha(x_3) \leq  \frac{\omega(x_3)}{2(\beta-1)}\, .
\end{equation}
Moreover, by virtue of \eqref{g-power}
$$ 
\int_r^\infty g(s)\, s\, \mathrm{d}s \, \leq \, \sqrt{B_0}\,
\int_r^\infty (1+\sqrt{B_0}\ s)^{1-2\beta}\, \mathrm{d}s = \frac{1}{2\beta-2}\, (1+\sqrt{B_0}\ r)^{2-2\beta} .
$$ 
Assume first that $1\leq k \leq \beta -1$. In this case a combination of \eqref{psi-k} and the last equation gives
\begin{eqnarray}
\lefteqn{\lambda_k(x_3) \leq \frac{\omega(x_3)\, B_0}{(\beta-1)\, \Gamma(k)} \,  \left(\frac{B_0}{2}\right)^{k} \int_0^\infty e^{-\frac{B_0 r^2}{2}} r^{2k-1} (1+\sqrt{B_0}\ r)^{2-2\beta}\, \mathrm{d}r} \nonumber \\
&& =   \frac{\omega(x_3)\, B_0}{(\beta-1)\, \Gamma(k)} \int_0^\infty e^{-s} s^{k-1}\, (1+\sqrt{2 s})^{2-2\beta}\, \mathrm{d}s \, \nonumber \\
&& \leq\,   \frac{\omega(x_3)\, B_0}{(\beta-1)\, \Gamma(k)} \, \int_0^\infty e^{-s} \, \mathrm{d}s  =  \frac{\omega(x_3)\, B_0}{(\beta-1)\, \Gamma(k)} \,, \phantom{AAAAAAAAAAAA} \label{k-geq-1}
\end{eqnarray}
where we have used again the substitution $s= \frac{B_0 r^2}{2}$.

\smallskip

\noindent On the other hand, for $k > \beta-1$ we have
\begin{eqnarray*}
\lefteqn{\lambda_k(x_3) \leq \frac{\omega(x_3)\, B_0}{(\beta-1)\, \Gamma(k)} \,  \left(\frac{B_0}{2}\right)^{k} \int_0^\infty e^{-\frac{B_0 r^2}{2}} r^{2k-1} (1+\sqrt{B_0}\ r)^{2-2\beta}\, \mathrm{d}r}  \\
&& \leq  \frac{\omega(x_3)\, B_0}{(\beta-1)\, \Gamma(k)} \,  \left(\frac{B_0}{2}\right)^{k} \int_0^\infty e^{-\frac{B_0 r^2}{2}} r^{2k-1} (B_0\, r^2)^{1-\beta}\, \mathrm{d}r \\
&& \leq  \frac{\omega(x_3)\, B_0}{(\beta-1)\, \Gamma(k)} \int_0^\infty e^{-s} s^{k-\beta}\,  \mathrm{d}s =  \frac{\omega(x_3)\, B_0\, \Gamma(k+1-\beta)}{(\beta-1)\, \Gamma(k)}\,. \phantom{AAAA}
\end{eqnarray*}
This together with equations \eqref{alpha-upperb}, \eqref{upperb-0}, \eqref{k-geq-1} and Lemma \ref{lem-aux} shows that
$$
 \sum_{k=0}^\infty\, \Lambda_k^\gamma(x_3)\, \leq \,  K(\beta, \gamma) \left(\frac{B_0}{\beta-1}\right)^{\gamma} \omega(x_3)^{\gamma}\, ,
$$
with the constant $K(\beta, \gamma)$ given by \eqref{eq-K}. The claim now follows from \eqref{eq-LT-gen-3d} upon setting $\gamma=\sigma +\frac 12$.
\end{proof}

\begin{proof}[\bf Proof of Theorem \ref{thm-gauss}]
We proceed as in the proof of Theorem \ref{thm-power} and use equation \eqref{upperb-0} and Lemma \ref{lem-aux}. Since
\begin{equation} \label{alpha-gauss}
\alpha(x_3)  \leq \omega(x_3)\, B_0 \int_0^\infty B_0\,  e^{- \varepsilon B_0 r^2} \, r\, \mathrm{d}r =  \frac{\omega(x_3)}{2\varepsilon}
\end{equation}
holds in view of \eqref{g-gauss}, for $k=0$ we get
$$
\Lambda_0(x_3) \, \leq \, \alpha(x_3) B_ 0 \, \leq \, \frac{\omega(x_3)\, B_0}{2\varepsilon}\,.
$$
On the other hand,
$$
\int_r ^\infty g(s)\, s\, \mathrm{d}s \leq  B_0 \int_r^\infty e^{- \varepsilon B_0 s^2} \, s\, \mathrm{d}s =  \frac{1}{2 \varepsilon}\, e^{- \varepsilon B_0 r^2} \, .
$$
Hence using the substitution $s= \frac{B_0 r^2}{2} (1+2\varepsilon)$, we obtain
\begin{eqnarray*}
\lefteqn{\lambda_k(z) \leq \frac{\omega(x_3)\, B_0}{\varepsilon\, \Gamma(k)} \,  \left(\frac{B_0}{2}\right)^{k} \int_0^\infty e^{-\frac{B_0 r^2}{2} (1+2\varepsilon)} \, r^{2k-1}\, \mathrm{d}r} \\
&& = \frac{\omega(x_3)\, B_0}{2\varepsilon} \, \frac{(1+2\varepsilon)^{-k}}{\Gamma(k)} \int_0^\infty e^{-s}\, s^{k-1}\, \mathrm{d}s  = \frac{\omega(x_3)\,  B_0}{2\varepsilon}\,  (1+2\varepsilon)^{-k}
\end{eqnarray*}
for any $k\geq 1$. Summing up gives
$$
\sum_{k=0}^\infty \, \Lambda_k^\gamma(x_3) \, \leq \,  \left(\frac{\omega(x_3)\, B_0}{2\, \varepsilon}\right)^\gamma \left(1+ \sum_{k=1}^\infty \left( (1+2\varepsilon)^{-k} +\frac{1}{2\sqrt{2\pi k}}\right)^\gamma \right)\,.
$$
Theorem \ref{thm-red} applied with $\gamma=\sigma +\frac 12$ then completes the proof.
\end{proof}

\smallskip
\begin{proof}[\bf Proof of Theorem \ref{thm-hole}]
In this case we have
$$
\alpha(x_3) = \omega(x_3)\, \frac{B_0 R^2}{2}\,.
$$
Inequality  \eqref{upperb-0} thus implies
$$
\Lambda_0(z) \, \leq \, \omega(z)\, \frac{B_0^2\, R^2}{2}\, .
$$
For $k\geq 1$ we note that in view of \eqref{g-hole}
$$
\int_r^\infty g(s)\, s\, \mathrm{d}s =
\left\{
\begin{array}{l@{\quad}l}
\frac 12\, (R^2-r^2)  & \quad r \leq R \, \\
& \\
0   &\quad r > R  \\
\end{array}
\right.
$$
Hence from \eqref{psi-k} and \eqref{lam-k-2} we conclude that
\begin{eqnarray*}
\lefteqn{\lambda_k(z)\, \leq\,  \frac{B_0^2\, R^2\, \omega(x_3)}{\Gamma(k)} \, \left(\frac{B_0}{2}\right)^{k} \int_0^R e^{-\frac{B_0 r^2}{2}} \, r^{2k-1}\, \mathrm{d}r}  \\
&&\, \leq\,  \frac{B_0^2\, R^2\, \omega(x_3)}{2 \Gamma(k)}  \int_0^{\frac{B_0 R^2}{2}} e^{-s} \, s^{k-1}\, \mathrm{d}s \\
&&\, \leq\,  \frac{B_0^2\, R^2\, \omega(z)}{2 k \Gamma(k)} \,  \left(\frac{B_0\, R^2}{2}\right)^k  = \frac{B_0\, \omega(x_3)}{\Gamma(k+1)} \,  \left(\frac{B_0\, R^2}{2}\right)^{k+1}\!\!, \quad k\in\mathbb{N}\,.
\end{eqnarray*}
This in combination with the above estimate on $\Lambda_0(x_3)$ and Lemma \ref{lem-aux} implies
$$
\sum_{k=0}^\infty \, \Lambda_k^\gamma(x_3) \ \leq \  \omega(x_3)^\gamma\, B_0^\gamma\, \left(\frac{B_0\, R^2}{2}\right)^\gamma \left(1+ \sum_{k=1}^\infty \left(\left(\frac{B_0\, R^2}{2}\right)^{k} \frac{1}{k!} + \frac{1}{\sqrt{2\pi k}}\right)^\gamma\, \right) \,,
$$
and the claim follows again by applying Theorem \ref{thm-red} with $\gamma=\sigma +\frac 12$.
\end{proof}

\subsection*{Acknowledgements}

The research was supported by the Czech Science Foundation (GA\v{C}R) within the project 14-06818S. D.B. acknowledges the support of the University of Ostrava and the project ``Support of Research in the Moravian-Silesian Region 2013''. H.~K. was supported by the Gruppo Nazionale per Analisi Matematica, la Probabilit\`a e le loro Applicazioni (GNAMPA) of the Istituto Nazionale di Alta Matematica (INdAM).
The support of MIUR-PRIN2010-11 grant for the project  ``Calcolo delle variazioni'' (H.~K.) is also gratefully acknowledged. T.W. was in part supported by the DFG project WE 1964/4-1 and the DFG GRK 1838.



\end{document}